\documentclass[conference]{IEEEtran}
\IEEEoverridecommandlockouts
\usepackage{cite}
\usepackage{amsmath,amssymb,amsfonts}
\usepackage{algorithmic}
\usepackage{graphicx}
\usepackage{mathrsfs}
\usepackage{textcomp}
\usepackage{xcolor}
\usepackage{amsthm}

\newtheorem{corollary}{Corollary}
\newtheorem{theorem}{Theorem}
\newtheorem{remark}{Remark}
\newtheorem{example}{Example}
\def\BibTeX{{\rm B\kern-.05em{\sc i\kern-.025em b}\kern-.08em
    T\kern-.1667em\lower.7ex\hbox{E}\kern-.125emX}}
\begin{document}

\title{Private Semantic Communications with Separate Blind Encoders\\
%
}

\author{\IEEEauthorblockN{1\textsuperscript{st} Amirreza Zamani}
\IEEEauthorblockA{\textit{School of Electrical Engineering and Computer Science} \\
	\textit{KTH Royal Institute of Technology, Stockholm, Sweden}\\
	Stockholm, Sweden,\\
	Email: amizam@kth.se}
\and
\IEEEauthorblockN{2\textsuperscript{nd} Mikael Skoglund}
\IEEEauthorblockA{\textit{School of Electrical Engineering and Computer Science} \\
	\textit{KTH Royal Institute of Technology, Stockholm, Sweden}\\
	Stockholm, Sweden,\\
	Email: skoglund@kth.se}
%
}

\maketitle

\begin{abstract}
We study a semantic communication problem with a privacy constraint where an encoder consists of two separate parts, e.g., encoder 1 and encoder 2. The first encoder has access to information source $X=(X_1,\ldots,X_N)$ which is arbitrarily correlated with private data $S$. The private data is not accessible by encoder 1, however, the second encoder has access to it and the output of encoder 1. 
A user asks for a task $h(X)$ and the first encoder designs the semantic of the information source $f(X)$ to disclose. Due to the privacy constraints $f(X)$ can not be revealed directly to the user and the second encoder applies a statistical privacy mechanism to produce disclosed data $U$. Here, we assume that encoder 2 has no access to the task and the design of the disclosed data is based on the semantic and the private data.

 In this work, we propose a novel approach where $U$ is produced by solving a privacy-utility trade-off based on the semantic and the private data. We design $U$ utilizing different methods such as using extended versions of the Functional Representation Lemma and the Strong Functional Representation Lemma. We evaluate our design by computing the utility attained by the user. Finally, we study and compare the obtained bounds in a numerical example.
\end{abstract}


\section{Introduction}
Semantic communication focuses on transmitting a modified version of the original information, tailored for the receiver to extract meaning relevant to a specific task or goal \cite{sherdeniz}. Unlike traditional methods, it accounts for the context, nuances, and connotations of the message, aiming to align with the recipient’s knowledge and expectations to avoid ambiguity \cite{feist2022significance,sherdeniz}. Recently, semantic communication has gained attention as a promising approach for reducing data load in 6G and future networks by transmitting only semantically relevant information. Data privacy is a critical aspect of emerging communication systems. As data usage expands, it is essential to limit the exposure of sensitive information. Semantic communication inherently addresses some privacy concerns by restricting the transmission of sensitive data. However, certain tasks may still require sharing information correlated with private data, which can be challenging to detect and protect.

Related works on the semantic communications and statistical approach to the privacy can be found in
\cite{strinati2024goal, zamani2024semantic, zamani2024multi, task,task2,task3,photis1, photis2, sherdeniz,feist2022significance,strinati20216g,shao2022theory,borz,khodam,Khodam22,kostala,shah, makhdoumi, dwork1, shahab,9457633,e25040679, liu2020robust,sankar, Total, sankar2, deniz4, asoodeh3, Calmon1,  nekouei2, issa,oof}.
\begin{figure}[]
	\centering
	\includegraphics[scale = .55]{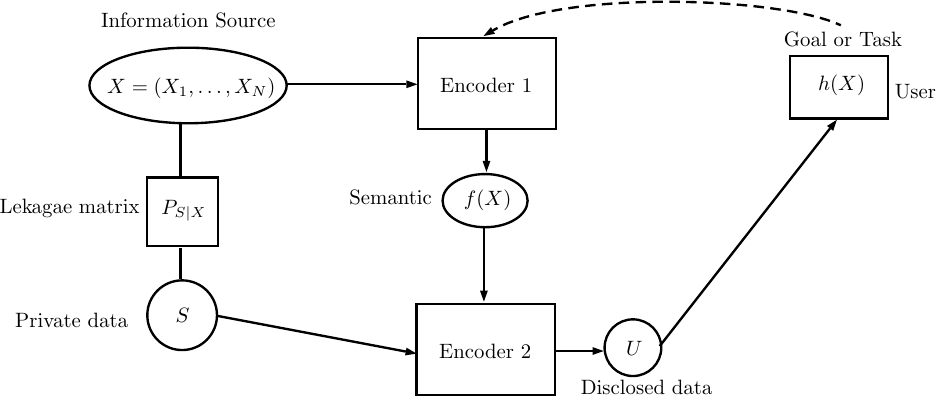}
	\caption{The encoder consists of two separate parties with different goals. The goal is to design $U$ that maximizes the utility while satisfying a certain privacy constraint.}
	\label{motiv}
\end{figure}
we have introduced a semantic communication problem with privacy constraints in \cite{zamani2024semantic} where a privacy-utility trade-off is obtained to produce released data. To solve the privacy-utility trade-off extended versions of the Functional Representation Lemma (FRL) and the Strong Functional Representation Lemma (SFRL) have been used. It has been shown that the released data achieve the optimal trade-off under certain assumptions. In \cite{zamani2024multi}, the private semantic communication considered in \cite{zamani2024semantic} is extended by considering a multi-task scenario and upper and lower bounds on the privacy-utility trade-off have been derived. A multi-user semantic communication problem using deep learning has been considered in \cite{task}. A deep learning based semantic communication model has been considered in \cite{task2} to address multiple transmission tasks. In \cite{task3}, an information bottleneck and adversarial learning approach is proposed to retain users’ privacy against model inversion attacks. An information theoretic approach to goal-oriented semantic communications is proposed in \cite{photis1} and \cite{photis2}, where in \cite{photis1} a rate distortion approach has been used to address a goal-oriented semantic information transmission problem. In \cite{photis2}, a joint source channel coding variation for identifying goal-oriented semantic aspects of messages over noisy channels with multiple distortion constraints has been considered.
Considering statistical privacy mechanism designs, in \cite{borz}, the problem of privacy-utility trade-off considering mutual information both as measures of privacy and utility is studied. Under perfect privacy assumption, it has been shown that the privacy mechanism design problem can be reduced to linear programming. 
In \cite{khodam}, privacy mechanisms with a per-letter privacy criterion considering an invertible leakage matrix have been designed allowing a small leakage. This result is generalized to a non-invertible leakage matrix in \cite{Khodam22}. In \cite{kostala}, \emph{secrecy by design} problem is studied under the perfect secrecy assumption. Bounds on secure decomposition have been derived using the Functional Representation Lemma. 
In \cite{shah}, the privacy problems considered in \cite{kostala} are generalized by relaxing the perfect secrecy constraint and allowing some leakage. Furthermore, more bounds have been derived in \cite{zamanistatistical} using extended versions of the FRL and the SFRL and `\emph{separation technique}'.

In this paper, we study a scenario illustrated in Fig. \ref{motiv}. Here, the encoder consists of two separate parties, encoder 1 and encoder 2. Encoder 1 has access to the information source but not the private data. Moreover, encoder 2 has access to the private data and not the information source.
We refer to the encoders as '\emph{blind}' since the first encoder has no access to the private data, and the second encoder is unaware of the task, respectively.
 The goal of encoder 1 is to design a semantic based on the task for the user and encoder 2 wants to design the disclosed data that satisfies the privacy constraint and keeps as much information as possible of the semantic. 
 Hence, it is obtained by solving a privacy-utility trade-off problem.   
  We find upper and lower bounds on the privacy-utility trade-off. To find lower bounds we use extended versions of the FRL and the SFRL and we obtain simple privacy mechanism designs. The proposed schemes offer mathematical approaches to design goal-oriented privacy mechanisms. These mechanisms not only facilitate the receiver in achieving its goal but also guarantee the privacy of the sensitive data from the recipient. To evaluate our design we compute the utility gained by the user and show that the distance between the upper and lower bounds on the user's utility can be small. Finally, we study the obtained bounds on the user's utility in a numerical experiment which is based on MNIST data set.

\section{System model and Problem Formulation}
Let $P_{SX}=P_{SX_1,\cdots,X_N}$ denote the joint distribution of discrete random variables $X=(X_1,\ldots,X_N)$ and $S$ defined on finite alphabets $\mathcal{X}=\mathcal{X}_1\times\cdots\times\mathcal{X}_N$ and $\cal{S}$. 
Here, $S$ denotes the latent private data, while $X$ is the information source of dimension $N$. The semantic of $X$ is defined as a function of $X$ denoted by $f(X):\mathcal{X}_1\times\cdots\times\mathcal{X}_N\rightarrow \mathbb{R}^T$ with dimension $T\leq N$. Furthermore, the goal or task of communication is represented by some other function of $X$, i.e., $h(X):\mathcal{X}_1\times\cdots\times\mathcal{X}_N\rightarrow \mathbb{R}^K$, with dimension $K\leq N$. 
In general $f\neq h$, since $f(\cdot)$ needs to be designed in an efficient way and efficiency can be defined based on different parameters. In this paper, encoder 1 designs the semantic which satisfies the following constraints:
\begin{align}
\gamma_1\leq I(f(X);h(X))&\leq \gamma_2 < H(h(X)),\label{varoo1}\\
h(f(X))&\leq \gamma_3.\label{varoo2}
\end{align}
The constraints $I(f(X);h(X))\leq \gamma_2$ and $h(f(X))\leq \gamma_3$ represent the limitations of the first encoder. As an example, $I(f(X);h(X))\leq \gamma_2$ can correspond to the limited bandwidth of the encoder. Specifically, $h(f(X))\leq \gamma_3$ corresponds to the bounded compression rate of the semantics. Additionally, the constraint $\gamma_1\leq I(f(X);h(X))$ represents the distortion condition of the semantics; for instance, the correlation between the task and the semantics must exceed a threshold defined based on the application. For more detail see \cite{sherdeniz} and \cite{photis1}. As we outlined earlier we refer to the encoder 1 \emph{blind} since it has no access to the private data and the design of the semantic is based on the information source and the task. Since the semantic is correlated with the private data, encoder 1 sends it to the second encoder. Encoder 2 designs the disclosed data described by RV $U\in \mathcal{U}$ based on the semantic and the private data. Similarly, we refer to encoder 2 \emph{blind} since it has no access to the task of communication. To design $U$, the goal of encoder 2 is to keep as much information as possible about the semantic while satisfying a privacy constraint. We use mutual information to measure the privacy leakage. Hence, encoder 2 obtain $U$ by solving the following privacy-utility trade-off
\begin{align}
h_{\epsilon}(P_{S,f(X)})&=\!\!\!\!\!\!\!\!\sup_{\begin{array}{c} 
	\substack{P_{U|S,f(X)}: I(U;S)\leq\epsilon,}
	\end{array}}\!\!\!\!\!\!\!\!I(f(X);U),\label{main1ch11}
\end{align} 
where the semantic $f(X)$ satisfies \eqref{varoo1} and \eqref{varoo2}, $P_{S,f(X)}$ is the joint distribution of $(S,f(X))$, and $P_{U|S,f(X)}$ describes the conditional distribution. Moreover, the utility achieved by the user can be measured by the mutual information between the disclosed data $U$ and the task $h(X)$, i.e., $I(h(X);U)$. In contrast to the present work, in \cite{zamani2024semantic}, we assumed that the encoder has access to the private data $S$, the semantic $f(X)$ and the task $h(X)$, and the design of the disclosed data $U$ is based on them. 
\begin{example}
	A scenario that motivates our model is to let $X$ be the outcome for the patients, encoder 1 be a laboratory, encoder 2 be a hospital and the user be a company which asks for a task based on the information source. For some security reason only the laboratory has access to the sensitive information.
	\end{example}
\begin{remark}
	\normalfont
	For $\epsilon=0$ and $f(X)=X$, \eqref{main1ch11} lead to the secret-dependent perfect privacy function $h_0(P_{S,X})$, studied in \cite{kostala}, where upper and lower bounds on $h_0(P_{SX})$ have been derived. In \cite{shah}, we have strengthened these bounds. Additionally, for $\epsilon\geq 0$ and $f(X)=X$, the problem $h_{\epsilon}(P_{S,X})$ has been studied in \cite{shah}, where upper and lower bounds are obtained. The privacy mechanism design is based on the extended versions of the FRL and the SFRL.
\end{remark}
\begin{remark}
	\normalfont
	A multi-task scenario semantic communication problem has been studied in \cite{zamani2024multi}, where a user has $L$ tasks which is a subset of the set $\{X_1,\ldots,X_N\}$. Similar to \cite{zamani2024semantic}, in \cite{zamani2024multi} we assumed that the encoder has access to the private data, the semantic and the task.
\end{remark}
\section{Main Results}
In this section, we obtain upper and lower bounds on \eqref{main1ch11} and the utility attained by the user $I(U;h(X))$.
Before stating the next theorem we derive an expression for $I(f(X);U)$. We have
\begin{align}
I(f(X);U)&=I(S,f(X);U)-I(S;U|f(X)),\nonumber\\&=I(S;U)+H(h(X)|S)\nonumber\\&-H(h(X)|U,S)-I(S;U|h(X)).\label{key}
\end{align}
Moreover, we expand the term $I(U;h(X),f(X))$ as follows
\begin{align}\label{key2}
I(U;h(X),f(X))&=I(U;f(X))+I(U;h(X)|f(X))\nonumber\\&=I(U;h(X))+I(U;f(X)|h(X)).
\end{align}
Next, we derive lower and upper bounds on $h_{\epsilon}(P_{S,f(X)})$ and the utility $I(h(X);U)$. For deriving lower bounds we use EFRL \cite[Lemma 4]{shah} and ESFRL \cite[Lemma 5]{shah}.
For simplicity the remaining results are derived under the assumption $f(\cdot):\mathcal{X}_1\times\ldots\mathcal{X}_N\rightarrow \mathbb{R}$, i.e., $T=1$. 
\begin{theorem}
	For any $0\leq \epsilon$, joint distribution $P_{S,f(X),h(X)}$ and any semantic $f(X)$ which satisfies \eqref{varoo1} and \eqref{varoo2}, we have
	\begin{align}\label{th2ch11}
	\!\!\!\max\{L_{h}^{1}(\epsilon),L_{h}^{2}(\epsilon)\}\!\leq\! h_{\epsilon}(P_{S,f(X)})\leq H(f(X)|S)+\epsilon,
	\end{align}
	where
	\begin{align}
	L_{h}^{1}(\epsilon) &= H(f(X)|S)-H(S|f(X))+\epsilon,\label{s1}\\
	L_{h}^{2}(\epsilon) &= H(f(X)|S)-\alpha H(S|f(X))+\epsilon\nonumber\\&-(1-\alpha)\left( \log(I(S;f(X))+1)+4 \right),\label{s2}
	\end{align}
	with $\alpha=\frac{\epsilon}{H(S)}$.
	The lower bound in \eqref{th2ch11} is tight if $H(S|f(X))=0$, i.e., $S$ is a deterministic function of $f(X)$. Furthermore, if the lower bound $L_{h}^{1}(\epsilon)$ is tight then we have $H(S|f(X))=0$. Additionally, for the utility attained by the user we have
	\begin{align}\label{util}
		\!\!\!\max\{L^{1}(\epsilon),L^{2}(\epsilon)\}\leq L^{3}(\epsilon)\!\leq\! I(h(X);U)&\leq H(f(X)|S)+\epsilon\nonumber\\&+H(h(X)|f(X)),
	\end{align}
	where 
	\begin{align}
	L^{1}(\epsilon) &= L_{h}^{1}(\epsilon)-H(f(X)|h(X)),\\
	L^{2}(\epsilon) &= L_{h}^{2}(\epsilon)-H(f(X)|h(X)),\\
	L^{3}(\epsilon) &= h_{\epsilon}(P_{S,f(X)})-H(f(X)|h(X)),
	\end{align}
	and $L_{h}^{1}(\epsilon)$ and $L_{h}^{2}(\epsilon)$ are defined in \eqref{s1} and \eqref{s2}.
\end{theorem}
\begin{proof}
	Using \eqref{key} we have
	$
	I(U;f(X))=I(U;S)+H(f(X)|S)-I(U;S|f(X))-H(f(X)|S,U),
	$
	which results in 
	$
	I(U;h(X))\leq \epsilon+H(f(X)|S).
	$
	For deriving the lower bounds $L_{h}^{1}(\epsilon)$ and $L_{h}^{2}(\epsilon)$ we use EFRL and ESFRL using $S\leftarrow X$ and $f(X)\leftarrow Y$. 
	Let $\bar{U}$ and $\tilde{U}$ be the output of the EFRL and the ESFRL. Using the same arguments in \cite[Theorem 2]{shah} we have
	\begin{align}
	I(\bar{U};f(X))&\geq L_{h}^{1}(\epsilon),\\
	I(\tilde{U};f(X))&\geq L_{h}^{2}(\epsilon),\\
	I(\bar{U};S)&=I(\tilde{U};S)=\epsilon.
	\end{align}
	For more detail see the proof for \cite[Theorem 2]{shah}. The main idea for constructing a RV $U$ that satisfies EFRL or ESFRL constraints is to add a randomized response to the output of the FRL or the SFRL. The randomization introduced in \cite{warner} is taken over $S$. The results about tightness can be proved by using \cite[Theorem 2]{shah}. To prove the upper bound in \eqref{util} let $U$ satisy the leakage constrain $I(U;S)\leq \epsilon$. Using \eqref{key2}, we have
	 \begin{align}
	 I(U;h(X))&\leq I(U;f(X))+I(U;h(X)|f(X))\nonumber\\& \leq I(U;f(X)) + H(h(X)|f(X))\nonumber\\ &\leq h_{\epsilon}(P_{S,f(X)})+ H(h(X)|f(X))\nonumber\\ &\stackrel{(a)}{\leq} H(f(X)|S)+\epsilon+ H(h(X)|f(X)),
	 \end{align}
	 where in step (a) we used the upper bound in \eqref{th2ch11}. To achieve the lower bounds $L^{1}(\epsilon)$ and $L^{2}(\epsilon)$, let $\bar{U}$ and $\tilde{U}$ achieve $L_{h}^{1}(\epsilon)$ and $L_{h}^{2}(\epsilon)$, respectively. Furthermore, let $U^*$ attain $h_{\epsilon}(P_{S,f(X)})$. Using \eqref{key2}, for any $U$ we have
	 \begin{align}\label{kir}
	 I(U\!;\!h(X))&\!\geq\! I(U\!;\!f(X))\!+\!I(U\!;\!h(X)|f(X))\!-\!H(f(X)|h(X))\nonumber\\&\geq I(U;f(X))-H(f(X)|h(X)). 
	 \end{align}
	 Substituting $U$ by $\bar{U}$ in \eqref{kir} we have
	 \begin{align*}
	 I(\bar{U};h(X))&\geq I(\bar{U};f(X))-H(f(X)|h(X))\\&=L_{h}^{1}(\epsilon)-H(f(X)|h(X))=L^1(\epsilon).
	 \end{align*}
	 Furthermore, substituting $U$ by $\tilde{U}$ in \eqref{kir} we obtain
	  \begin{align*}
	  I(\bar{U};h(X))&\geq I(\tilde{U};f(X))-H(f(X)|h(X))\\&=L_{h}^{2}(\epsilon)-H(f(X)|h(X))=L^2(\epsilon).
	  \end{align*}
	  Finally, substituting $U$ by $U^*$ in \eqref{kir} we have
	  \begin{align*}
	  I(\bar{U};h(X))&\geq I(U^*;f(X))-H(f(X)|h(X))=L^3(\epsilon).
	  \end{align*}
	\end{proof}
\begin{remark}
	Although $L^3(\epsilon)\geq \max \{L^1(\epsilon),L^2(\epsilon)\}$, computing $L^3(\epsilon)$ can be challenging since we need to find $h_{\epsilon}(P_{S,f(X)})$. In some cases, the upper bound in \eqref{th2ch11} is tight, but not in general. The lower bounds $L^1(\epsilon)$ and $L^2(\epsilon)$ can be easily computed and the privacy mechanisms that attain them are simple since the extended versions of the FRL and the SFRL have constructive proofs.
\end{remark}
\begin{example}
	Let the private data $S$ be a deterministic function of the semantic $f(X)$, i.e., $H(S|f(X))=0$. In this case, $h_{\epsilon}(P_{S,f(X)})=H(f(X)|S)+\epsilon=L_{h}^{1}(\epsilon)$ and we have
	\begin{align*}
	L^3(\epsilon)=L^1(\epsilon)=H(f(X)|S)+\epsilon-H(f(X)|h(X)).
	\end{align*}
	Furthermore, the distance between the upper and lower bounds on $I(U;h(X))$ is
	$
	H(f(X)|h(X))+H(h(X)|f(X)).
	$
	When the semantic is highly correlated with the task the gap can be small.
\end{example}
\begin{remark}
	More lower bounds on $h_{\epsilon}(P_{S,X})$ have been obtained in \cite{zamanistatistical} by combining the extended versions of the FRL and the SFRL, and `\emph{separation technique}'. Due to the complexity of the lower bounds obtained in \cite{zamanistatistical}, we used only $L_{h}^{1}(\epsilon)$ and $L_{h}^{2}(\epsilon)$. 
\end{remark}
Next, we show that tightness of the upper bound in Theorem 1 can be improved using the concept of \emph{common information}. In other words, instead of having $H(S|f(X))=0$ we propose larger set of distributions that the upper bound is attained. To do so we use the definition of the common information between $X$ and $Y$ as introduced in \cite{wyner}.
\begin{corollary}\label{sefff}
	For any $0\leq \epsilon$, if the common information and mutual information between the private data $S$ and semantic $f(X)$ are equal then we have
	$
	h_{\epsilon}(P_{S,f(X)})=H(f(X)|S)+\epsilon.
	$
\end{corollary} 
\begin{proof}
	The proof is based on \cite[Theorem 3]{shah} and \cite[Proposition 3]{shah}.
\end{proof}
Corollary 1 improves the condition $H(S|f(X))=0$ for achieving the upper bound in Theorem 1.
Consequently, when the semantic $f(X)$ is a deterministic function of $S$ i.e., $H(f(X)|S)=0$, we have
$
h_{\epsilon}(P_{S,f(X)})=\epsilon.
$
Using Corollary 1, we have
\begin{align}
&H(f(X)|S)+\epsilon-H(f(X)|h(X))\leq I(h(X);U)\\&\leq H(f(X)|S)+\epsilon+H(h(X)|f(X)).
\end{align} 
Next, we find an upper bound on the utility attained by the user $I(U;h(X))$ which does not depend on the semantic.
\begin{corollary}
	For any $0\leq \epsilon$, joint distribution $P_{S,f(X),h(X)}$ and any semantic $f(X)$ which satisfies \eqref{varoo1} and \eqref{varoo2}, we have
	\begin{align}
	I(U;h(X))\leq \epsilon+\gamma_2-\gamma_1+H(h(X)).
	\end{align}
\end{corollary}
\begin{proof}
	The proof is straightforward since by using Theorem 1 we have
	\begin{align*}
	I(U;h(X))\leq& \epsilon + H(f(X)|S)+H(h(X)|f(X))\\ \leq&\epsilon+H(h(X)) + \gamma_2-\gamma_1,
	\end{align*}
	where in last line we used \eqref{varoo1} and \eqref{varoo2}.
\end{proof}
\section{Numerical Experiment}
In this section, we provide an example to evaluate the bounds obtained in Theorem 1. To do so, we use the MNIST data set which contains $60000$ images illustrating handwritten digits from $0$ to $9$. Let RV $X$ (information source) denote the images in the data set, i.e., $|\mathcal{X}|=6000$. 
Let 
$h(X)$ represent the task that is a feature extracted from the data set. In this example, the task is defined as a "quantized histogram." To determine the quantized histogram of the data set, we first convert the images into black and white by quantizing them. The histogram of any image in the data set is then computed as the ratio of the number of white pixels to the total number of pixels. Consequently, the histogram of any image can be represented by a value within the interval $[0,1]$. 
For quantizing the histogram of any image, we use the following intervals: Interval~1 $=[0,0.1)$, Interval~2 $=[0.1,0.15)$, Interval~3 $=[0.15,0.2)$, Interval~4 $=[0.2,0.25)$, Interval~5 $=[0.25,0.3)$, Interval~6 $=[0.3,0.35)$ and Interval~7 $=[0.35,1]$. Each image in the data set is therefore labeled with a number between 1 and 7, indicating the interval to which its histogram belongs. Let $h(X)$ denote the labels corresponding to the intervals. Clearly, the labels to the intervals are deterministic functions of the information source $X$. 
In this example, the private attribute denoted by a binary RV $S$, indicates the presence of digit $5$. Furthermore, let RV $Z$ denote the label of each image that is a number from $0$ to $9$ indicating the digit which the image illustrates, i.e., $Z\in\{0,1,...,9\}$. We can choose $Z$ to be the semantic and clearly it is a deterministic function of $X$, i.e., $Z=f(X)$. 
Since $Z$ contains the information about the sensitive data, it cannot be transmitted directly. 
Here, we use empirical distribution for $Z$ and it can be seen that $S$ is a deterministic function of $Z$ which yields the Markov chain $S-Z-h(X)$. Using the Markov chain $S-Z-h(X)$, the kernel $P_{S|h(X)}$ can be obtained 
and $P_{Z|h(X)}$ can be found from the data set. We use empirical distribution to find the kernel $P_{h(X)|Z}$, e.g., $P_{h(X)=1|Z=1}$ is the number of images illustrating digit $1$ which corresponds to label $1$ (their quantized histogram belong to the first interval) divided by total number of images with digit $1$. We sweep $\epsilon$ to evaluate our lower and upper bounds derived in Theorem 1. In Fig.~\ref{ajab3}, lower bounds $L^{1}(\epsilon)$ and $L^{2}(\epsilon)$ and upper bound $H(f(X)|S)+\epsilon+H(h(X)|f(X))$ are illustrated for different leakages. In this figure, we compare the bounds. We can see that the first lower bound is dominant and the gap between $L^{1}(\epsilon)$ and the upper bound is $H(h(X)|f(X))+H(f(X)|h(X))\simeq 1.4$ nats. To plot the second lower bound we have used $[L^2(\epsilon)]^+$ where $[x]^+=\max\{0,x\}$. Furthermore, the second lower bound approaches the first bound as the leakage increases. When $\epsilon=H(S)$ we have $\alpha=1$ which yields $L^{1}(\epsilon)=L^{2}(\epsilon)$.
\begin{figure}[]
	\centering
	\includegraphics[scale=0.17]{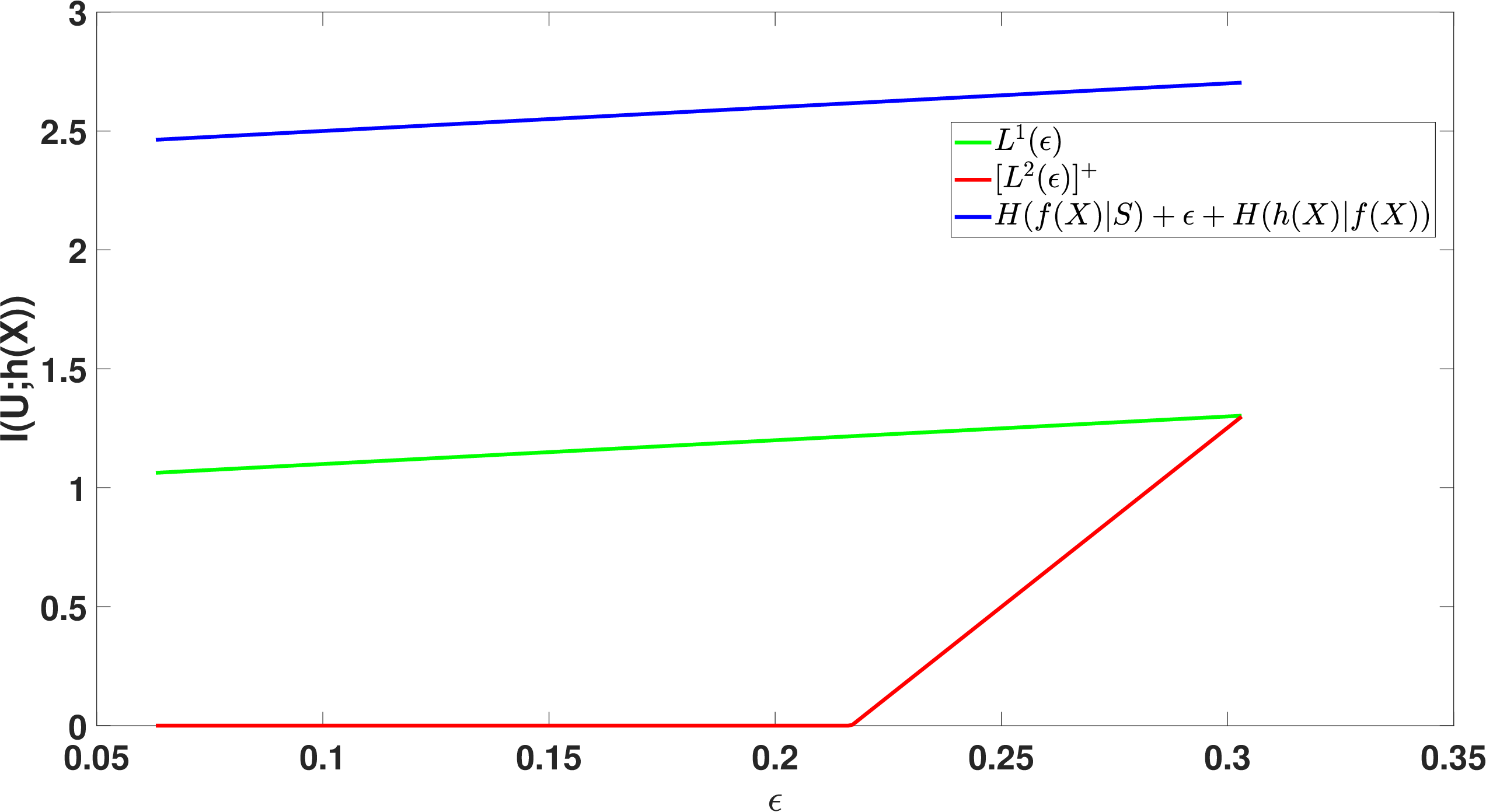}
	\caption{Comparing the lower and upper bounds obtained in Theorem 1. In this example, the gap between the upper bound and the first lower bound is $H(h(X)|f(X))+H(f(X)|h(X))\simeq 1.4$ nats.}
	\label{ajab3}
\end{figure}
\vspace{-6mm}
\section{conclusion}\label{concull}
We have introduced 
a semantic communication with privacy constraint with two separate and blind encoders. It has been shown that using extended versions of the FRL and the SFRL lower bounds on $h_{\epsilon}(P_{S,f(X)})$ and the utility attained by the user $I(U;h(X))$ are obtained. 
When the private data $S$ is a deterministic function of the semantic $f(X)$, the upper bound on $h_{\epsilon}(P_{S,f(X)})$ is achieved. Also, this statement is generalized by using the concept of common information.  
Finally, we have studied the bounds considering different scenarios. The gap between the upper and lower bounds on the utility $I(U;h(X))$ is studied considering different scenarios and it has been shown that under some constraints it can be small.

\newpage
\bibliographystyle{IEEEtran}
\bibliography{IEEEabrv,IZS}
\vspace{12pt}

\end{document}